\documentclass[AMA,STIX1COL]{WileyNJD-v2}
\usepackage{NJDnatbib}
\usepackage{float}
\usepackage{graphicx}
\usepackage{caption}
\usepackage{subcaption}
\graphicspath{ {plots1/} }

\newtheorem{example}[theorem]{Example}

\articletype{Article Type}%

\received{\today}
\revised{\today}
\accepted{\today}

\begin{document}

\title{Network Topology Design to Influence the Effects of Manipulative Behaviors in a Social Choice Procedure}

\author[1]{Athanasios-Rafail Lagos*}

\author[1,2]{George P. Papavassilopoulos}

\authormark{Athanasios-Rafail Lagos \textsc{et al}}

\address[1]{\orgdiv{Department of Electrical and Computer Engineering}, \orgname{National Technical University of Athens}, \country{Greece}}

\address[2]{\orgdiv{Ming Hsieh Department of Electrical and Computer Engineering -Systems}, \orgname{University of  Southern California}, \country{US} }

\corres{*Athanasios-Rafail Lagos \email{lagosth993@gmail.com}}

\abstract[Summary]{A social choice procedure is modeled as a Nash game among the social agents. The agents are communicating with each other through a social communication network modeled by an undirected graph and their opinions follow a dynamic rule modelling conformity. The agents' criteria for this game are describing a trade off between self-consistent and manipulative behaviors. Their best response strategies are resulting in a dynamic rule for their actions. The stability properties of these dynamics are studied. In the case of instability, which arises when the agents are highly manipulative, the stabilization of these dynamics through the design of the network topology is formulated as a constrained integer programming problem. The constraints have the form of a Bilinear Matrix Inequality (BMI), which is known to result in a nonconvex feasible set in the general case. To deal with this problem a Genetic Algorithm, which uses an LMI solver during the selection procedure, is designed. Finally, through simulations we observe that in the case of topologies with few edges, e.g. a star or a ring, the isolation of the manipulative agents is an optimal (or suboptimal) design, while in the case of well-connected topologies the addition or the rewiring of just a few links can diminish the negative effects of manipulative behaviors.}


\maketitle

\section{Introduction}
In recent years great progress has been made in the mathematical modeling and study of social phenomena. A topic of current interest is the study of the evolution of social agents' opinions about a certain issue. The knowledge of the mechanisms of the formation and the propagation of the agents' opinions are very useful in several fields. For example, in marketing the advertisers care about the opinion of the consumers for the advertised product and in politics the politicians care about the opinion of the agents about their agenda. Thus, a lot of work has been done in this field  \cite{de_groot},\cite{Friedkin1},\cite{Friedkin2},\cite{Sznajd},\cite{h&k1},\cite{h&k2},\cite{Fortunato1},\cite{Krause},\cite{extremist},\cite{Deffuant},
\cite{Gionis},\cite{stubborn},\cite{acemoglu2010spread},\cite{forster2014trust}, many interesting cases have been modeled and analysed, some of which are summarized in \cite{Survey}, \cite{Friedkin_book}, and new ideas continue to be proposed and studied up to now \cite{Etesami1},\cite{Etesami2},\cite{Etesami3}.

The majority of these works \cite{de_groot}-\cite{stubborn} consider a single state for the agents, modelling their opinion, belief or attitude about an issue, and they study the dynamics of this state. The dominant mechanism that determines the evolution of the opinions is considered to be the averaging of the opinions of the agent's peers. The reason for this modelling are the tendencies of an agent to imitate her peers and to conform to her social group attitudes, which are both well-studied social phenomena. In fact, this modelling of opinion dynamics has been verified to be realistic by experimental data of a field research in India \cite{chandrasekhar2012testing}.

However, in many cases, such as social choice procedures (e.g. elections, referendums, polls), the organisation who studies the opinion dynamics cares to predict or to affect the outcome of this procedure, which is determined by the agents actions or behaviors. So, the question whether an agent's opinion imply a specific behavior-action naturally arises. The answer that the field of social psychology gives to this question is negative, in many cases the opinions do not imply specific actions \cite{lapiere1934attitudes},\cite{wicker1969attitudes},\cite{glasman2006forming}. Behavior is not solely dependent on one's beliefs but is drastically affected by the situations and in some cases behavior affects ones attitudes and beliefs \cite{marsh2005influence}.

In addition to moral and situational factors, game theory suggests that an agent's behavior is also dependent on her desire to maximize her private interests \cite{von2007theory}. So, the action-behavior of an agent is also shaped by her utility gained form the outcome of the social choice procedure, which also depends on the other agents' actions. This indicates that an agent's action depends on her neighbors' actions and it is a best response to them. This perspective adds the useful insight that the agents usually act antagonistically to their neighbors and they do not just conform to their peers' pressure \cite{Kordonis1}.

An advantage of the game theoretic modelling for the agents' actions is that it can explain better the emergence of manipulative behaviors in social choice procedures, which is a topic of significant interest. Several recent studies on several countries like U.S. \cite{benkler2018network},\cite{aral2019protecting} and Argentina \cite{szwarcberg2015mobilizing} indicate that social networks have become an arena of manipulative behaviors \cite{woolley2018computational}. Paid brokers of political parties, fake accounts (bots), echo chambers, organised disinformation (fake news, slandering) are some of the manipulation techniques that have arisen in the fertile ground of the online political conversations. Furthermore, in this new environment of political struggle each agent may act in a manipulative way in an effort to pull the social outcome to her favor, however, she may be less manipulative than an expert of the previous categories. Such behaviors are considered in some recent works \cite{acemoglu2010spread}-\cite{Etesami3}.

In this work, we extend a model introduced in \cite{Etesami3} describing a social choice procedure, where the population structure is modeled by an undirected graph and the agents' actions depend both on their opinions, which evolve dynamically in our model, and on their neighbors' actions. Specifically, we consider that each agent has an internal belief or opinion, which evolves in time in a way modelling a tendency of conformity to the public opinion. Each agent has also an expressed action in the social choice procedure. Each opinion matches to a proper action. However, the action of each agent isn't identical to her proper action, but it derives from the minimization of a criterion modeling the tendency of the agent to manipulate, i.e. to deviate from her proper action in order to pull the social outcome to her favor.

The resulting game between the agents is considered to be repeated in discrete time steps. The action shaping criteria of each agent retain the same form at each step. So, we formulate a series of one-step games where we seek for the Nash strategy profiles. These strategy profiles result in a dynamic rule for the actions of the agents coupled with the opinion dynamics. It is interesting that in the case that the agents are highly manipulative these dynamics become unstable, since the social outcome stands as a tug of war among the agents who try to pull it to their side.

Motivated form this fact, we study the stability properties of these dynamics and we deduce a sufficient condition that guarantees the convergence of the system to a bounded state. This condition implicates the manipulative tendencies of the agents and the graph structure with the stability of the system, stating that the acceptable manipulative behavior of an agent is relative to her position in the graph. Simulations are also presented in order to examine how the opinion and action dynamics behave over several well known graph structures, such as random graphs, lattices and small world graphs.

Subsequently, we consider the problem of changing the social network's topology in order to influence the effects of manipulative behaviors. The network topology has been chosen as our designing parameter for two basic reasons. At first, the network topology is a parameter that the social network's administrator can affect and thus influence the agents' behaviors in an indirect way, which may lead to less effort and costs than the enforcement of strict rules to the users of the network. Secondly, the network topology design is an emerging problem in many scientific fields nowadays, such as security \cite{Kordonis2}, communication networks \cite{chen2007network}, sensor networks \cite{ferentinos2007adaptive}, \cite{kar2008sensor}, distributed optimization \cite{gross2011optimized}, distributed LANs  \cite{khan2012fuzzy}, \cite{saad2018multi},  UAVs navigation \cite{uavs}, cyberphysical systems \cite{cyberphysical}, convergence of mean field games \cite{Kordonis3} etc.

So, a general formulation and study of a network topology design for the stabilisation of a system of unstable dynamics of interconnected agents can be applied to many practical problems of current interest. It must be specified that in our work the term social network corresponds to its digital realisation and not to its abstract concept of a representation of human relationships, so an administrator exists and the topology can be affected. We would like to note here that in contrast to its practicality the existence of one or more administrators in such networks raises the more intriguing question of who will control the administrators, who have the power to affect the other agents' manipulability and the final outcome.

For the topology design procedure, we study the case of an initial topology resulting in unstable dynamics and we want to find a new topology that results in stable dynamics and that is close to the initial topology with respect to the number and the exact position of their edges. This problem is formulated as an integer programming problem with a Lyapunov inequality for discrete time systems (known also as Schur's inequality) as constraint. Each decision variable of this optimisation problem represent either the existence of an edge between two agents or one of the components of the Lyapunov matrix. The constraint is nonlinear with respect to our decision variables and it can be written as a Bilinear Matrix Inequality, which is known to be a nonconvex problem in its general case \cite{Safonov}. A similar approach involving integer optimisation with a Bilinear Matrix Inequality constraint for the graph topology design problem has been presented in a similar design problem \cite{gross2011optimized}, where the authors considered a LMI relaxation of the problem and a branch and bound technique to deal with the integer decision variables.

In this work, we develop a genetic algorithm to deal with this problem. This algorithm searches only for the values of the integer decision variables representing the edges of the graph, while a Linear Matrix Inequality solver is used to check the feasibility of each new topology by solving the Lyapunov inequality with the topology variables fixed, which results to be linear with respect to the symmetric matrix of the Lyapunov function. This procedure is repeated for many generations, where new topologies are produced by the application of the genetic operators.

Finally, simulations of the results of the proposed algorithm are presented. The behavior of the algorithm is studied over several different initial topologies, where the agents' parameters have been chosen properly to raise instabilities in the dynamics. Through the examination of these test cases we derive conclusions on the functionality of the proposed algorithm and the relevancy of our results with the expected ones from our empirical perception of social networks and social choice procedures.

\section{Problem formulation}

\subsection{Notation}
We consider an undirected graph $G=(V,E)$. By $n$ we denote the number of the vertices of the graph, which represent the agents. We denote by $N_i$ the neighborhood of the agent $i$, $N_i=\{j:(i,j) \in E\}$ and by $d_i$ the degree of node $i$, that is the size of its neighborhood. Let $A$ be the adjacency matrix of the graph, it is a $n\times n$ symmetric matrix and its $(i,j)$ entry is 1 if nodes $i$ and $j$ are adjacent to each other and 0 otherwise. Let $D=\textrm{diag}\{d_i\}$ be the diagonal degree matrix, $C=diag\{c_i\}$ be a diagonal matrix of the self-confidence parameters $c_i$ and $G=diag\{g_i\}$ be a diagonal matrix of the manipulability parameters $g_i$. The symbol $\mathbb{1}$ stands for the $n\times 1$ vector with all its coordinates equal to 1. The symbol $\mathcal{I}$ stands for the identity $n\times n $ matrix and the symbols $e_i, i=1...n$ stand for the standard basis of $\mathbb{R}^n$ . For a set $A$ we denote $\mathcal{X}_A$ its indicator function, i.e. $\mathcal{X}_A(x)=1$ if $x\in A$ and $\mathcal{X}_A(x)=0$ elsewhere. The symbolism $\lceil \cdot \rceil$ denotes rounding to the next natural number and the symbolism $\lceil \cdot \rceil_{even}$  denotes rounding to the next even natural number. The space of the square $n\times n$ symmetric positive definite matrices is denoted $\mathcal{M}^{S+}_n$. The symbol $A^T$ stands for the transpose of the matrix $A$ and the symbol $\lambda_i(A)$ denotes the i-th eigenvalue of $A$. All the norms $\|\cdot\|$ that have no subscript stand for the 2-norm.

\subsection{Derivation of the Opinion Dynamics}
At first, the mechanism that determines the evolution of the agents' opinions is studied. The opinions, beliefs or attitudes of the agents represent an ideological state, that expresses what they believe about an issue and not what they actually do. The opinion/attitude of the agent $i$ is denoted by $\theta_i(k)$ at each time step $k$, and its value is a real number. In field researches, attitudes are usually measured in a five point scale, however, we consider here a continuous and unbounded analogue, which is common in the opinion dynamics literature.

It is considered that the main factors that shape the opinions in time are imitation and conformity. That is, the agents' opinions tend to be affected with their neighbors' opinions through continuous dialogue and finally reach a consensus. This model of opinions' evolution is well known and studied for many years \cite{de_groot}-\cite{Friedkin2}, \cite{stubborn}. In fact, in \cite{Friedkin2}, \cite{stubborn} the model has been enriched with the inclusion of stubborn agents, i.e. people who insist on their initial beliefs, but since their presence affects primarily the equilibrium of the opinion dynamics and not their stability properties, we shall not include such agents in our model. So, every agent has an initial opinion $\theta_i(0)$ and she changes her opinion at each time step according to following dynamic rule:

\begin{equation}
\theta_i(k+1)=\frac{c_i}{d_i+c_i}\theta_i(k)+\frac{1}{d_i+c_i}\sum_{j\in N_i}\theta_j(k)
\end{equation}
where $c_i$ is a factor analogue to the self-confidence of the agent for her opinion.

\subsection{Derivation of the Action Dynamics}
The actions of the agents represent what they actually do, in our case what they choose in the social choice procedure. The action/behavior of each agent is denoted by $u_i(k)$ at each time step $k$ and its value is a real number. As with the opinions, the actions could also be modeled to take values in a discrete scale, however, in this work we consider a continuous relaxation of that more difficult problem.

In contrast with the opinions which are shaping by a progressive conformity to the average beliefs, the criteria determining the action of each agent in every time step depict the tendency of the agents to manipulate the social outcome to their favor. That is, each agent may deviate her action from the one dictated by her beliefs in order to pull the social outcome towards her desired direction. In other words, as pointed out in \cite{Etesami2}, it is a common phenomenon in politics that the people who disagree with what they perceive as the expected social outcome tend to overstate their opinions, leading their neighbors to misperceptions of the public opinion and conform to these false estimations, thus pulling the social outcome to their favor. For this reason, an important parameter of their criteria is their estimation of the social outcome, based on their available information.

\begin{assumption}
It is assumed that the agents have local information of the other agents' actions, that is they know only the actions of their neighbors.
\end{assumption}

\begin{assumption}\label{as2}
It is assumed also that the information pattern is Markovian, i.e. at each time step they know only the last actions of their neighbors forgetting the past.
\end{assumption}
So, the available information for each agent is:
\begin{equation}\label{info}
  I_i(k)=\{\theta_i(k), u_j(k-1), \forall j\in N_i\}
\end{equation}
According to this information pattern the estimated social outcome for each agent is her local average, evaluated on the available samples at time $k$:
\begin{equation}
\tilde{u}_i(k)=\frac{\sum_{j\in N_i}u_j(k-1)+u_i(k)}{d_i+1}
\end{equation}
Based on the aforementioned concepts the criteria that determine the actions of each agent are dependent on her current opinion and on her locally estimated social outcome, so they are defined at each time step as follows:
\begin{align}
J^{a}_i(I_i(k))= (u_i-\phi_i(\theta_i))^2+g_i(\tilde{u}_i-\phi_i(\theta_i))^2
\end{align}
where $\phi_i(.)$ is a continuous transformation matching each agent's opinion to a desired behavior-action. The first term of the cost function $(u_i(k)-\phi_i(\theta_i(k)))^2$ indicates the self-consistency of the agent, i.e. how close her action is to an action consistent with her opinion, while the second term $g_i(\tilde{u}_i-\phi_i(\theta_i(k)))^2$ indicates the manipulative/opportunistic ends of the agent, i.e. how much she cares to affect the social outcome through her action so as to bring it close to her desirable outcome. The parameters $g_i$ determine the ratio between self-consistent and manipulative behaviour for each agent.

\begin{remark}
If the second assumption \ref{as2} is relaxed by adding memory to the agents, so as to be able to predict the social outcome based on all the previous actions of their neighbors, the one-step Nash game examined here will be converted to a dynamic one. The dynamic game is of high complexity, so assuming that the social agents have bounded rationality and they do not seek to solve a difficult problem to determine their social behavior, we deal with the one-step Nash game which is tractable.
\end{remark}

Assuming that the agents choose their actions rationally based on their criteria we seek for the Nash equilibrium solution of the one step game. These best-response actions derive from the solution of the following system of equations:
\begin{equation}\label{gi_eq}
\big\{\frac{\partial J^{a}_i}{\partial u_i}=0 \big\}
\end{equation}
which have the following form:

\begin{align*}
\frac{\partial J^a_i}{\partial u_i}=0 \Rightarrow 2(u_i-\phi(\theta_i))+2g_i\bigg(\frac{\sum_{j\in N_i}u_j+u_i}{d_i+1}-\phi_i(\theta_i)\bigg)\frac{1}{d_i+1}=0
\end{align*}
solving these equations with respect to $u_i$ and using the information pattern $I_i(k)$ to evaluate each quantity in accordance with the available sample at time $k$ we obtain the following dynamics for the actions:

\begin{align}\label{u_tel}
u_i(k+1)&=\big(1+\frac{d_ig_i}{g_i+(d_i+1)^2}\big)\phi_i(\theta_i(k+1))-\frac{g_i}{g_i+(d_i+1)^2}\sum_{j\in N_i}u_j(k)
\end{align}
Defining now the diagonal matrices
\begin{align}
G_{\theta}=diag\big\{1+\frac{d_ig_i}{g_i+(d_i+1)^2}\big\}
\end{align}
and
\begin{align}
G_u=diag\big\{\frac{g_i}{g_i+(d_i+1)^2}\big\}
\end{align}
we define the matrix
\begin{equation}
  A_u=G_uA
\end{equation}
and rewrite the equation \eqref{u_tel} in matrix form:

\begin{equation}\label{u_matrix}
u(k+1)=G_{\theta}\Phi(\theta(k+1))-A_uu(k)
\end{equation}
where $u(k)=[u_1(k)...u_n(k)]^T$ and $\Phi(\theta(k+1))=[\phi_1(\theta_1(k+1))...\phi_n(\theta_n(k+1))]^T$.

\section{Stability Analysis}

\subsection{Known results on opinion dynamics}\label{OD}
For the evolution of the opinions of the agents we consider the following dynamics:
\begin{equation}
\theta_i(k+1)=\frac{c_i}{d_i+c_i}\theta_i(k)+\frac{1}{d_i+c_i}\sum_{j\in N_i}\theta_j(k)
\end{equation}
which can be summarized using matrix notation
\begin{equation}
A_{\theta}=(D+C)^{-1}(A+C)
\end{equation}
to the following expression:
\begin{equation}
\theta(k+1)=A_{\theta}\theta(k)
\end{equation}
where $\theta(k)=[\theta_1(k)...\theta_n(k)]^T$ and $A_{\theta}$ is a row-stochastic, aperiodic matrix. So, $\theta(k)$ converges to a limit $\theta^c$ which is actually a consensus on each connected subgraph. For some results on these the reader could study \cite{de_groot} and for a more general description one could study the criteria summarised in \cite{Tsitsiklis}. So the following statements hold:
\begin{align}\label{theta_cauchy}
\|\theta(k+1)-\theta(k)\|\rightarrow 0 \\
\|\theta(k)-\theta^c\|\rightarrow 0
\end{align}

\subsection{Stability analysis of the coupled opinion and action dynamics}
We continue our analysis by considering the augmented vector $z(k)=[\theta_1(k)...\theta_n(k), u_1(k)...u_n(k)]^T$ and the resulting augmented system describing its dynamics. For simplicity of the presentation we will use the notation $\Phi\circ A_{\theta}\theta(k)$ to denote the nonlinear function $\Phi(\theta(k+1))$ . So we obtain the following dynamics:

\begin{equation}\label{dyn}
z(k+1)=
\left[\begin{array}{cc} A_{\theta} & 0\\ G_{\theta}\Phi\circ A_{\theta} & -A_u \end{array}\right]z(k)
\end{equation}

\begin{lemma}
If there exists a symmetric, positive definite matrix $P$ such that $A_u^TPA_u-P<0$ and the function $\Phi$ is continuous in $\mathcal{R}^n$ and locally Lipschitz in a neighborhood of $\theta^c$ with a Lipschitz constant $L_{\Phi}$, then the coupled dynamics \eqref{dyn} have an equilibrium which is globally asymptotically stable.
\end{lemma}
\begin{proof}
At first, we define the $P$-norm of a vector $x$: $\|x\|_P:=\sqrt{x^TPx}$ and of a matrix $A$: $\|A\|_P:=\textrm{sup}_{\{\|x\|_P=1\}}\{\|Ax\|_P\}$. From these definitions we have that if $A_u^TPA_u-P<0$ holds then $\|A_u\|_P<1$.\\

For the opinion dynamics, $\theta(k+1)=A_{\theta}\theta(k)$, it is known to be stable as we have already discussed in a previous section. So, $ \exists K: \forall k>K$ $\theta(k)$ belongs to a neighborhood of $\theta^c$ where the mapping $\Phi$ is Lipschitz. Thus $\forall k>K$ the following holds for the actions:

\begin{align}
\|u(k+1)-u(k)\|_P&=\|G_{\theta}\Phi(\theta(k+1))-A_uu(k)-G_{\theta}\Phi(\theta(k))+A_uu(k-1) \|_P \nonumber \\
&\leq \|G_{\theta}\Phi(\theta(k+1))-G_{\theta}\Phi(\theta(k))\|_P +\|A_uu(k)-A_uu(k-1)\|_P \nonumber \\
&\leq L_{\Phi}\|G_{\theta}\|_P\|\theta(k+1)-\theta(k)\|_P +\|A_u\|_P\|u(k)-u(k-1)\|_P
\end{align}
let $a=\|A_u\|_P<1$, as we have assumed and $\delta_k=L_{\Phi}\|G_{\theta}\|_P\|\theta(k+1)-\theta(k)\|_P\rightarrow 0 $ due to \eqref{theta_cauchy} and the fact that $\|x\|_P\leq \sqrt{\lambda_{max}(P)}\|x\|$.Thus, defining $x_k=\|u(k)-u(k-1)\|_P$,  we rewrite the previous inequality:
\begin{equation}
x_{k+1}\leq ax_k+\delta_k
\end{equation}
with $a<1$ and $\frac{\delta_k}{1-a} \rightarrow 0 $. So, this inequality satisfy the conditions of lemma 3, p.45 of \cite{Polyak} and consequently it converges to zero, thus the sequence $\|u(k+1)-u(k)\|_P$ is convergent to zero, so the sequence $u(k)$ is convergent to an equilibrium point. So finally, the coupled dynamics have an equilibrium which is globally asymptotically stable.
\end{proof}
The usefulness of this lemma arises form the fact that the opinion dynamics are stable for every graph structure as the matrix $A_{\theta}=(D+C)^{-1}(A+C)$ has the desired properties for every adjacency matrix $A$ and its degree matrix $D$. So, this lemma enables us to focus on the stabilization of the action dynamics, through the graph design and the consequent tuning of the matrix $A_u$, guaranteeing that the coupled dynamics will remain stable for every such design.

From the previous lemma, using $P=\mathcal{I}$ in the Lyapunov matrix inequality (thus $\|\cdot\|_P=\|\cdot\|_2$) and $G_uA=G_u(D+\mathcal{I})(D+\mathcal{I})^{-1}A$ we can derive the following simple but restrictive stability condition for the spectral radius of $G_u(D+\mathcal{I})$, $\rho(G_u(D+\mathcal{I}))=max\{|\lambda_i(G_u(D+\mathcal{I}))|, i=1...N\}$ to be less than one as well or equivalently:
\begin{equation}\label{suf}
\frac{(d_i+1)g_i}{g_i+(d_i+1)^2}\leq 1 \Rightarrow g_i\leq d_i+2, \forall i
\end{equation}
since it this case
\begin{align*}
  \|A_u\| \leq \|G_u(D+\mathcal{I})\|\|(D+\mathcal{I})^{-1}A\| \leq max\{|\lambda_i(G_u(D+\mathcal{I}))|, i=1...N\}\|(D+\mathcal{I})^{-1}A\|
\end{align*}
because the matrix $G_u(D+\mathcal{I})$ is diagonal. For the second norm it holds:
\begin{align*}
  \|(D+\mathcal{I})^{-1}A\| \leq \|(D+\mathcal{I})^{-1}\|\|A\| =\frac{1}{d_{max}+1}\|A\|
\end{align*}
and for $\|A\|$ it holds
\begin{equation*}
  \|A\|\leq\sqrt{\|A\|_{\infty}\|A\|_1}=d_{max}.
\end{equation*}
So,
\begin{equation*}
  \|(D+\mathcal{I})^{-1}A\|\leq \frac{d_{max}}{d_{max}+1}<1.
\end{equation*}

\begin{remark}\label{rem3}
We state this simple observation here because we can exploit its simplicity to use it as a heuristic for a stable topology design. That is, since this condition guarantees that the coupled dynamics converge on a graph with $min\{d_i\}\geq max\{g_i\}-2$ we know that a ring lattice of degree $d_1=\lceil max\{g_i\}-2\rceil_{even}$ is a topology that stabilizes these dynamics.
\end{remark}

\subsection{Simulations on the model's stability properties}
We present here some simulations of the aforementioned dynamics over different graph structures, that motivated us to formulate the topology design problem. In these simulations we consider a network of $n=20$ agents participating in a repeated social choice procedure for $T=100$ times. The parameters $c_i$ indicating the obstinateness of the agents are chosen from the interval $[10,100]$. The parameters $g_i$ indicating the manipulative tendencies of the agents are randomly chosen from the interval $[0,15]$. Their initial opinions are randomly chosen from the $[0,10]$ interval. Their initial actions are the desired ones according to their initial opinions $u_i(0)=\phi(\theta_i(0))$, where the function $\Phi$ is considered to be $\Phi(\theta)=10\textrm{tanh}(\theta/10)$, which is both continuous and locally Lipschitz.

Firstly, we present the convergent opinion and action dynamics Fig:\ref{fig:ra_04} on a realization of a random graph \cite{random} with edge probability $p=0.4$ Fig:\ref{fig:ra_04_g}. In the presented case the graph has $|E|=81$ edges and the spectral radius of the resulting matrix $A_u$ equals $\lambda_{max}\{A_u\}=0.7774$, so it has the necessary stability properties.

\begin{figure}[h!]
	\centering
	\begin{subfigure}{.45\textwidth}
        \includegraphics[width=\textwidth]{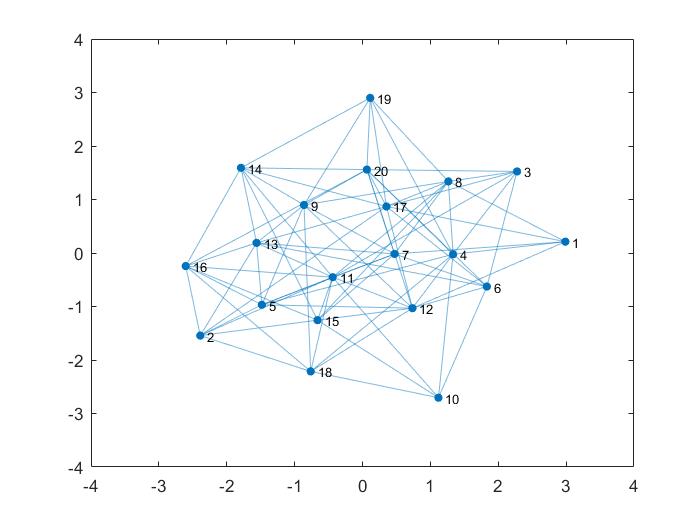}
        \caption{A random graph with edge probability p=0.4.}
        \label{fig:ra_04_g}
    \end{subfigure}
	\begin{subfigure}{.45\textwidth}
		\includegraphics[width=\textwidth]{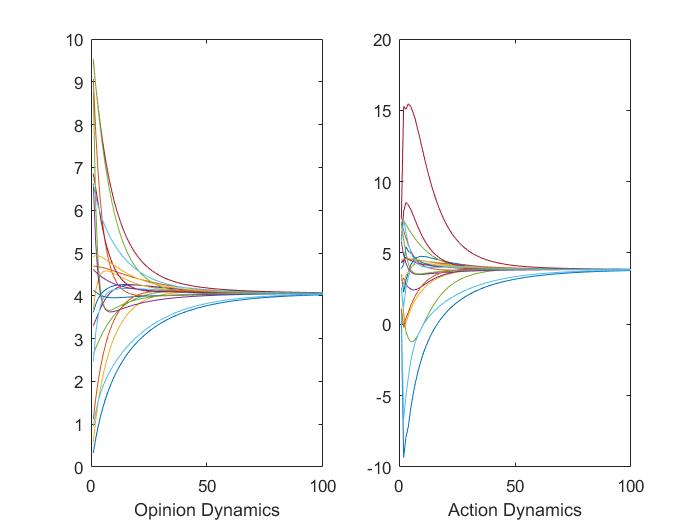}
		\caption{Opinion and action dynamics.}
        \label{fig:ra_04}
	\end{subfigure}
\end{figure}

Subsequently, a case of nonconvergent dynamics will be presented. The dynamics Fig:\ref{fig:unstable} result from a realization of a random graph with edge probability $p=0.3$ Fig:\ref{fig:unstable_g}, which has $|E|=54$ edges and  $\lambda_{max}\{A_u\}=1.0418$.

\begin{figure}[h!]
	\centering
	\begin{subfigure}{.45\textwidth}
        \includegraphics[width=\textwidth]{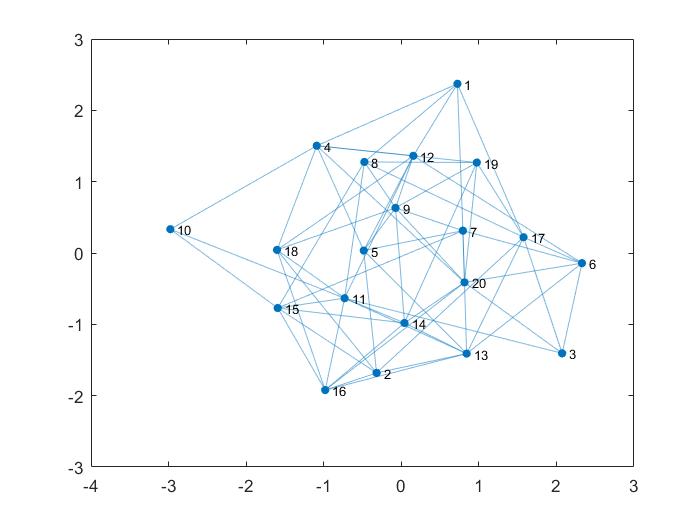}
        \caption{A random graph with edge probability p=0.3.}
        \label{fig:unstable_g}
    \end{subfigure}
	\begin{subfigure}{.45\textwidth}
		\includegraphics[width=\textwidth]{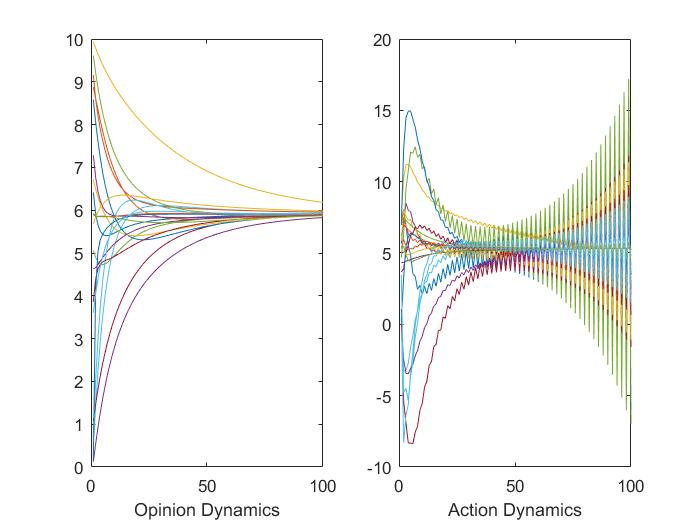}
		\caption{Opinion and unstable action dynamics.}
        \label{fig:unstable}
	\end{subfigure}
\end{figure}

We consider now the problem of choosing a proper graph structure, which would result in stable dynamics and be as close as possible to the aforementioned unstable one with respect to the edge number $|E|$ in this case. We make several experiments beginning from an $L^*$-lattice (a graph where all the agents have the same degree $L^*$), which satisfies our sufficient condition ($L^*>g_{max}-2$), $L^*=14$ in this example. Then we relax this condition by considering lattices of smaller node degree until the dynamics become unstable, as shown in table \ref{table}.

\begin{table}[h!]
\centering
\begin{tabular}{l l l}
\toprule
\textbf{Graph structure} & \textbf{$\lambda_{max}(A_u)$} & \textbf{$|E|$}\\
\midrule
$L^*$-lattice & 0.4042 & 140 \\
8-lattice & 0.7758 & 80 \\
6-lattice & 1.0114 & 60 \\
Small-world & 0.9491 & 60 \\
\bottomrule
\end{tabular}
\caption{Stability of several graph structures}
\label{table}
\end{table}
The most interesting observation we made from our experiments was that while the 6 degree lattice results in unstable dynamics if we rewire some of its edges and thus create a small world graph, as introduced by J. Watts and S.Strogatz (1998), the dynamics become stable. This indicates that a well structured topology -whose properties can be studied analytically- is not necessarily a best choice for our problem, on the opposite the introduction of some random rewirings results in better structures. This was a motivation for the following general formulation of the topology design problem, which is not restricted on several special classes of topologies.

\section{Network topology design for the stabilization of the action dynamics}

\subsection{Notation and Problem statement}
The network topology design problems are emerging in many different fields \cite{Kordonis2}-\cite{Kordonis3} and in more formulations they are considered to be difficult (NP-hard) problems. That is because the decision variables stand for the presence, the addition or the removal of nodes or links and so they take integer values, resulting in Integer Programs with various types of constraints.

Similarly, in our case we consider the vector $x\in \{0,1\}^{n(n-1)/2}$, which denotes the occurrence of a change of an edge -addition or removal of an edge- in the existing graph structure and constitutes our decision variables. The nodes of the graph remain unchanged.

Let $\{P^k, k=1...\frac{n(n+1)}{2}\}$ be a basis of the symmetric $n\times n$ matrices. Specifically, consider the matrices $P^k$ with $P^k_{ij}=P^k_{ji}=1$ if $i=\max_{m\geq 0}\{\sum_{l=1}^{m-1}(n-(l-1))\leq k\}$ and $j=i-1+k-\sum_{l=1}^{i-1}(n-(l-1))$ and $P^k_{ij}=0$ elsewhere. The diagonal matrices of this basis, i.e. $\{P^k: k\in K_d=\{\sum_{l=1}^{i-1}(n-(l-1))+1, i=1...n\}\}$, we will denote them $P^i_d$ since each $k\in K_d$ corresponds to an $i \in \{1...n\}$.

\begin{example}
  We present for example the aforementioned basis for the $2\times 2$ symmetric matrices:
  \begin{align*}
    P^1 & = \begin{bmatrix} 1 & 0\\ 0 & 0 \end{bmatrix} & P^2 & = \begin{bmatrix} 0 & 1\\ 1 & 0 \end{bmatrix} & P^3 & = \begin{bmatrix} 0 & 0\\ 0 & 1 \end{bmatrix}
  \end{align*}
  The set $K_d=\{P^1, P^3\}$, so $P^1_d=P^1$ and $P^2_d=P^3$.
\end{example}

Now with this notation we can write $A_0=\sum_{k\notin K_d}x_0(k)P^k$, where the vector $x_0$ stands for the coordinates of $A_0$ with respect to the aforementioned basis $\{P^k, k=1...\frac{n(n+1)}{2}\}$ except its diagonal elements whose coordinates are all zero. From the definition of $P^k$ it holds that $x_0(k) \in \{0,1\}$.

The topology design procedure consists of the addition of some new edges and the removal of some existing edges. So, we define the following sign function $\mathcal{S}_{x_0}(k)=1$ if $x_0(k)=0$ and  $\mathcal{S}_{x_0}(k)=-1$ if $x_0(k)=1$, which multiplied with the vector of changes $x$ indicates which changes correspond to an addition of an edge and which to a removal.

So the adjacency matrix of the graph depends linearly on the changes' vector $x$:
\begin{equation}
A(x)=A_0+\sum_{k=1}^{n(n-1)/2}x(k)P^k\mathcal{S}_{x_0}(k)
\end{equation}
form this equation we deduce that $A(x)=[\mathcal{L}^A_{ij}(x)]$  where $\mathcal{L}^A_{ij}(x)$ are linear functions of $x$. The degree matrix changes accordingly:
\begin{equation}
D(x)=\sum_{i=1}^ne_i(A(x)\mathbb{1})^TP^i_d
\end{equation}
which also depends linearly on $x$, i.e. $D(x)=\textrm{diag}\{\mathcal{L}^D_i(x)\}$ where $\mathcal{L}^D_i(x)$ are linear functions of $x$.

Subsequently, we define the matrix functions:
\begin{align}
G_u(x)=G(G+(D(x)+I)^2)^{-1}=\textrm{diag}\{\frac{g_i}{g_i+(\mathcal{L}^D_i(x)+1)^2}\}
\end{align}
and
\begin{equation}
A_u(x)=G_u(x)A(x)
\end{equation}
which are nonlinear with respect to the decision variables $x$.

Applying the Lyapunov stability criterion on the matrix $A_u(x)=G_u(x)A(x)$ we obtain the following matrix inequality for $P>0$ and $x$:
\begin{equation}
A(x)G_u(x)PG_u(x)A(x)-P\leq Q
\end{equation}
The matrix $Q$ is a negative definite matrix, for example $Q=-q\mathcal{I}$, where $q$ is a design parameter affecting the stability properties of the system as well as the size of the feasible region of the optimisation problem. In the simulations presented in the next section this parameter takes values of the order: $q\sim10^{-2}$.

Let the $F_x=\{x: \exists P>0: A(x)G_u(x)PG_u(x)A(x)-P\leq-q\mathcal{I}\}$. This set contains all the feasible designs, i.e. the vectors $x$ for witch the induced graph described by the adjacency matrix $A(x)$ has the desired stability properties.

In order to choose an element of the aforementioned feasible set as a best design, we consider the criterion of the minimum change from the initial graph structure, which is a natural criterion as especially on graphs representing social interactions it may be very difficult to persuade someone to abandon a friend or make a new one. So we consider the minimization of $ \|x\|_1 $, which is equivalent to the minimization of the linear objective $\textbf{1}^Tx$. The resulting problem is:

\begin{align}\label{problem}
\min_{x,P} &\{ \textbf{1}^Tx \}\\
& x\in \{0,1\}^{n(n-1)/2} \\
& \exists P > 0: A(x)G_u(x)PG_u(x)A(x)-P\leq-q\mathcal{I} \label{last_c}
\end{align}

\begin{remark}
  If for some reasons some edges of the network are considered to be more important than others, or the cost to add or remove them is different, we can formulate a similar optimisation problem substituting the objective by a weighted sum of the changes $w^Tx$, $w_i\geq 0$. Moreover, several linear constraints may be added so as to describe restrictions on the design parameters due to special structural characteristics of the network, which may be important to be preserved or due to special characteristics of several nodes, whose neighborhood cannot be affected. These changes in the optimisation problem formulation do not increase the difficulty of the problem as it lies on the constraint \eqref{last_c}.
\end{remark}

In order to simplify the nonlinear, non-polynomial (on the decision variables $x$) constraint $\exists P>0: A(x)G_u(x)PG_u(x)A(x)-P\leq-q\mathcal{I}$ we consider the change of variables $Z=G_u(x)PG_u(x)$ and prove the following proposition.
\begin{proposition}
For every point $x$, if there exists a matrix $Z > 0$:
\begin{equation}\label{z-con}
  A(x)ZA(x)-G_u^{-1}(x)ZG_u^{-1}(x)\leq-q\mathcal{I}
\end{equation}
then there exists a matrix $P > 0$: $$A(x)G_u(x)PG_u(x)A(x)-P\leq-q\mathcal{I}$$.
\end{proposition}
\begin{proof}
We use the mapping $Z=G_u(x)PG_u(x)$ from $P\in \mathcal{M}^{S+}_n$ to $Z\in \mathcal{M}^{S+}_n$.
For each element of the matrices $Z$ and $P$ it holds that :
$$z_{ij}=\frac{g_ig_j}{[g_i+(\mathcal{L}^D_i(x)+1)^2][g_j+(\mathcal{L}^D_j(x)+1)^2]}p_{ij}, $$
which is a bijection. Moreover, if $Z > 0$ then for $P=G_u^{-1}ZG_u^{-1}$ it holds that for every vector $x$:
$$x^TPx=x^TG_u^{-1}ZG_u^{-1}x=v^TZv >0$$ for $v=G_u^{-1}x$, so $P > 0$.
Finally, substituting the change of variables $Z=G_u(x)PG_u(x)$ in $A(x)ZA(x)-G_u^{-1}(x)ZG_u^{-1}(x)\leq-q\mathcal{I}$ we take the desired inequality $A(x)G_u(x)PG_u(x)A(x)-P\leq-q\mathcal{I}$.
\end{proof}

The new constraint \eqref{z-con} is polynomial in the decision variables $x$, so with a proper change of variables it can be transformed to a Bilinear Matrix Inequality (BMI). We give the following simple example, from \cite{vanantwerp2000tutorial} p.372, to explain this change of variables:
\begin{example}
  Let the polynomial inequality $x^3+yz<1$ . Defining $w=x^2$ and $v=x$ we have the following equivalent system of bilinear inequalities:
  \begin{align*}
    1-xw-yz & >0  \\
    w-xv & \geq 0 \\
    xv-w & \geq 0 \\
    x-v & \geq 0 \\
    v-x & \geq 0
  \end{align*}
\end{example}
In our case, each element $h_{ij}$ of the polynomial matrix $A(x)ZA(x)-G_u^{-1}(x)ZG_u^{-1}(x)$ is a $4^{th}$ degree polynomial of the decision variables $x$:
\begin{align*}
  h_{ij}=  \sum_{l=1}^{n}(\sum_{k=1}^{n}\mathcal{L}^A_{ik}(x)z_kl)\mathcal{L}^A_{lj}(x)-\frac{z_{ij}}{g_ig_j)}[g_i+(\mathcal{L}^D_i(x)+1)^2)][g_j+(\mathcal{L}^D_j(x)+1)^2)]
\end{align*}
Using the fact that $(x(k))^n=x(k)$ for every $n$ since $x(k)$ is 0 or 1 and introducing some extra variables $y_{kl}=x(k)x(l)$ we can write the polynomial matrix inequality \eqref{z-con} as a  Bilinear Matrix Inequality, with the aid of the matrices of the basis $\{P^k\}$.

The feasibility of a BMI is known to be a nonconvex problem in its general case \cite{Safonov}, so the same holds for our initial problem \eqref{problem}-\eqref{last_c}. The difficulty to deal with the BMI integer constrained problem is also discussed in \cite{gross2011optimized}. Moreover, due to the difficulty of the topology design problem in general, it has to be stated here that almost all of our references in this topic \cite{Kordonis2}-\cite{Kordonis3} use heuristics or meta-heuristics, except the ones considering simplifying assumptions or relaxations to deal with a convex problem in the end.

\subsection{A genetic algorithm for the topology design problem}

Genetic algorithms are a well known meta-heuristic which can be applied to obtain suboptimal solutions in a variety of difficult (NP-hard) search and optimisation problems \cite{goldberg1989genetic}. As such, it is evident that these algorithms are a useful tool for dealing with network topology design problems and they have already been applied in this field \cite{ferentinos2007adaptive},  \cite{kumar1995genetic}. Following this direction, we develop a genetic algorithm to obtain a feasible solution for the nonconvex integer programming problem \eqref{problem}-\eqref{last_c}. In order to avoid the explosion of the dimensionality which results to a very slow convergence for the genetic algorithm, we use the genetic algorithm to search only in the space of the decision variables $x$ rather than in the whole space $(x,P)$. However, this search may lead to several topologies which will not satisfy the constraint \eqref{last_c}. To deal with this we observe that the constraint \eqref{last_c} is linear with respect to the matrix variable $P$, so its feasibility can be efficiently checked with the use of a projective method based algorithm for Linear Matrix Inequalities (LMIs). So, for each new topology produced by the genetic operations we check its feasibility with an LMI solver and we drop it out of the next generation if it is infeasible.  The basic characteristics of this algorithm are enlisted below:

\textbf{Chromosomes:} Each chromosome of the genetic algorithm is a 0-1 sequence of length $\frac{n(n-1)}{2}$ representing the vector $x_0+x\cdot \mathcal{S}_{x_0}$ for some changes' vector x. The vectors $x_0$, x and $\mathcal{S}_{x_0}$ are defined in the previous section, while the symbol "$\cdot$" denotes elementwise multiplication of the two vectors.

\textbf{Initial population:} As initial population for the genetic algorithm we consider a specific number of feasible random perturbations of the initial topology $x_0$. That is we produce a number of chromosomes of the form $x_0+x\cdot \mathcal{S}_{x_0}$, which satisfy the constraint \eqref{last_c}, where x are randomly derived 0-1 sequences. The feasibility check, which is described below, is applied on these chromosomes in order to verify which of them are satisfying the constraint \eqref{last_c} and reject the others from the initial population.

\textbf{Fitness function:} The fitness function of the genetic algorithm coincides with the objective function of the problem  \eqref{problem}-\eqref{last_c}, so it has the following form $\textrm{fitness}(\textrm{chromosome})=\|\textrm{chromosome}-x_0\|_1=\|x_0+x\cdot \mathcal{S}_{x_0}-x_0\|_1=\|x\|_1 $.

\textbf{Selection:} For the choice of a portion of the population for the breeding of the next generation we use a simple truncation selection criterion. We choose the $50\%$ fittest part of the population in the case the size population exceeds a specific lower bound or we hold the whole population if its size is smaller than this lower bound. The reason for this is to avoid the diminishment of the population in the case that many new offsprings are rejected because they do not satisfy the constraints. The next generation of the population is initialised by the selected part of the previous population.

The truncation selection has the drawback that it may lead to elitism, that is the selection of only the temporarily best chromosomes which may be far from the global optimum. Thus, the algorithm may converge to a local minimum of the optimisation problem, but the convergence speed of the algorithm if we use another selection procedure, such as fitness proportionate selection, is much slower, so we have kept this simple method for our experimental simulations. Moreover, by choosing our initial conditions relatively close to the optimum - we initialise the algorithm with perturbations of the initial infeasible topology which are adequately close to it and feasible - we enhance our chances to find the global optimum even with this selection procedure. Of course, in cases of practical interest where great accuracy is needed and with sufficient computing power available, we can easily replace this subroutine by one applying fitness proportionate selection.

\textbf{Crossover:} The crossover operator considered here chooses randomly two parents form the selected portion of the population and chooses also randomly a crossover point between $1...\frac{n(n-1)}{2}$ and produces two offsprings form the two possible combinations of the parent chromosomes around this point.

\textbf{Mutation:} The mutation operator applied to an offspring changes each of its bits with probability $p_m=\frac{2}{n(n-1)}$, resulting on an average change of one bit per chromosome.

\textbf{Feasibility check:} After the production of the new offsprings with the application of the genetic operators, each offspring is checked for the feasibility of the constraint \eqref{last_c}. For this we use an LMI solver, which uses a projective method algorithm, to examine the existence of a matrix $P>0$ which satisfies the LMI \eqref{last_c}, where the matrices $A(x)$ and $G_u(x)$ have the fixed values corresponding to the vector x of the offspring's chromosome $x_0+x\cdot \mathcal{S}_{x_0}$. If this LMI is found feasible the new chromosome is added to the next generation, else it is rejected.

\textbf{Termination criterion:} The genetic algorithm terminates after a specified number of generations $N$. In fact, in the following simulations we have chosen the number of generations through experimentation so as to not observe any improvement in the objective function in the final generations. The fittest chromosome of the last generation is returned as solution for our topology design problem.

\subsection{Simulations of the results of the genetic algorithm}
In the following simulations we consider a network of $n=20$ agents participating in a repeated social choice procedure for $T=300$ times. The parameters $c_i$ are chosen randomly from the interval $[10,100]$. The parameters $g_i$ are randomly chosen form the interval $[0,10]$. The function $\Phi$ which maps the opinions to the desired actions is considered to be $\Phi(\theta)=10tanh(\theta/10)$, which is both continuous and locally Lipschitz. The initial opinions $\theta_i(0)$ are randomly chosen from the interval $[0,10]$ and the initial actions are the ones corresponding to these opinions $u_i(0)=\phi(\theta_i(0))$. All the aforementioned parameters remain the same in both simulations.

The initial graph topology is the realisation of a random graph with edge probability $p=0.2$ shown in figure Fig:\ref{fig:initial}. The resulting opinion and action dynamics are shown in figure Fig:\ref{fig:ini_dy}, where we can see that the action dynamics are unstable.

\begin{figure}[h!]
	\centering
	\begin{subfigure}{.45\textwidth}
        \includegraphics[width=\textwidth]{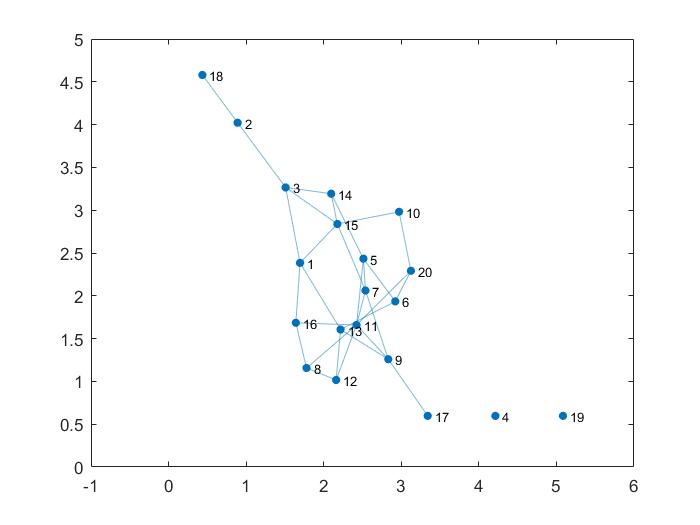}
        \caption{The initial graph topology, derived as a random graph with edge probability $p=0.2$.}
        \label{fig:initial}
    \end{subfigure}
	\begin{subfigure}{.45\textwidth}
		\includegraphics[width=\textwidth]{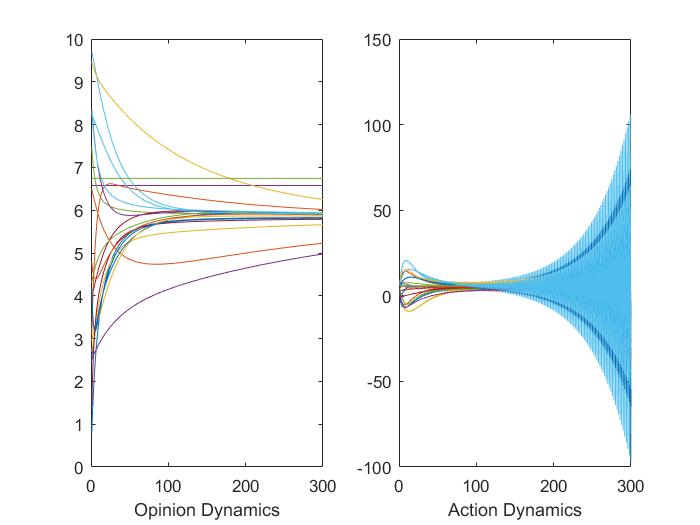}
		\caption{Unstable action dynamics on the initial graph topology.}
        \label{fig:ini_dy}
	\end{subfigure}
\end{figure}

Applying the genetic algorithm presented in the previous section to the initial graph topology we obtain the graph topology presented in figure Fig:\ref{fig:design}, which differs from the initial one only on three edges. The resulting opinion and action dynamics are shown in figure Fig:\ref{fig:final_dy}, where we can see that the action dynamics are stable over the designed graph topology.

\begin{figure}[h!]
	\centering
	\begin{subfigure}{.45\textwidth}
        \includegraphics[width=\textwidth]{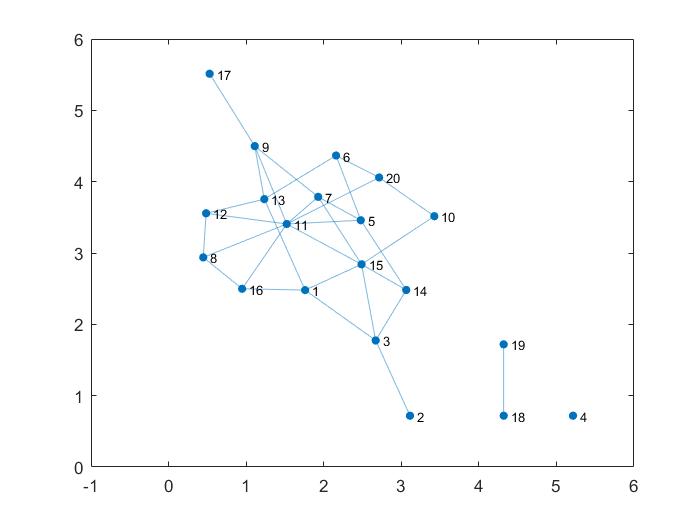}
        \caption{The designed graph topology by the genetic algorithm.}
        \label{fig:design}
    \end{subfigure}
	\begin{subfigure}{.45\textwidth}
		\includegraphics[width=\textwidth]{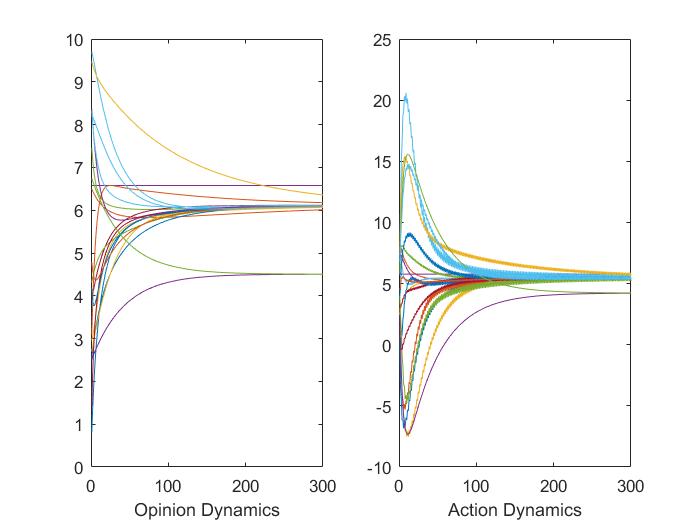}
		\caption{Stable action dynamics on the designed graph topology.}
        \label{fig:final_dy}
	\end{subfigure}
\end{figure}

\textbf{Comments:} As we observe from the simulations above the graph topology that derived from the genetic algorithm, as a feasible point of our optimisation problem which satisfies the BMI constraint, results in stable action dynamics. Moreover, with respect to its optimality, we have already pointed that the designed topology differs from the initial one on just 3 edges (specifically 1 edge has been removed and 2 new edges have been added), meaning that $\|x\|_1=3$ which is very small. It may be a suboptimal solution, but in most cases it might be an acceptable design. Finally, compared with the heuristic approaches developed in section 4.3 it outperforms them  by far, since the best we had achieved there was a difference of 8 on the amount, not on the exact location, of the existing edges of the two topologies, while now we achieved a difference of 3 on the exact location of the edges of the two topologies.

\subsection{Simulations over Special Structured Initial Topologies}

In the following simulations we consider a network of $n=20$ agents and we check just the structure of the resulting topologies after the implementation of the genetic algorithm on several special structured initial topologies. The parameters $g_i$ indicating the manipulative tendencies of the agents are chosen accordingly in each case in order to make the initial topology resulting in unstable dynamics.

\subsubsection{Ring}

For the ring topology Fig:\ref{fig:ring} the parameters $g_i$ indicating the manipulative tendencies of the agents are chosen randomly from the interval $[0,10]$. The ring is a very sparse structure for a connected one. It has only 20 edges while 19 are needed in order to be connected. Furthermore, its stability properties are not very enhanced - even small manipulative parameters result in instabilities. So, a connected stable topology differs a lot from the initial one. That's why our algorithm returns an unconnected topology as the optimal solution Fig:\ref{fig:unconnected_ring}. This topology has 5 edges and differs from the initial one on 15 edges. The unconnected designed topology is stable, since the isolation of the agents pauses their social interactions and results in the preservation of their initial opinions and actions, which are stable by the time they are not increasing.
\begin{figure}[h!]
	\centering
	\begin{subfigure}{.3\textwidth}
        \includegraphics[width=\textwidth]{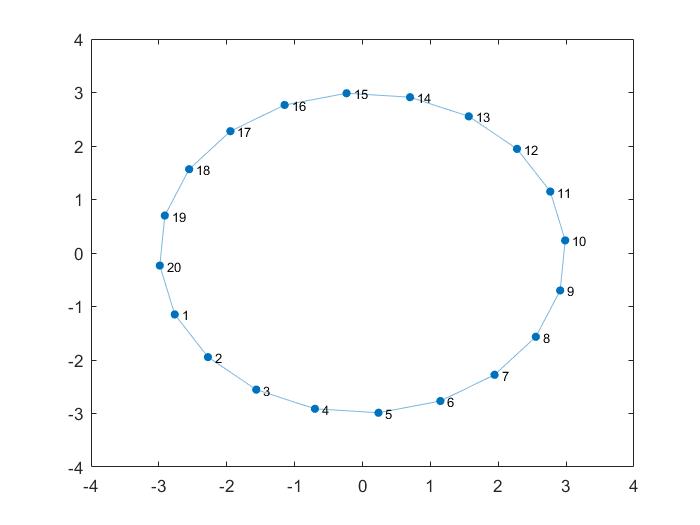}
        \caption{Initial ring topology}
        \label{fig:ring}
    \end{subfigure}
	\begin{subfigure}{.3\textwidth}
		\includegraphics[width=\textwidth]{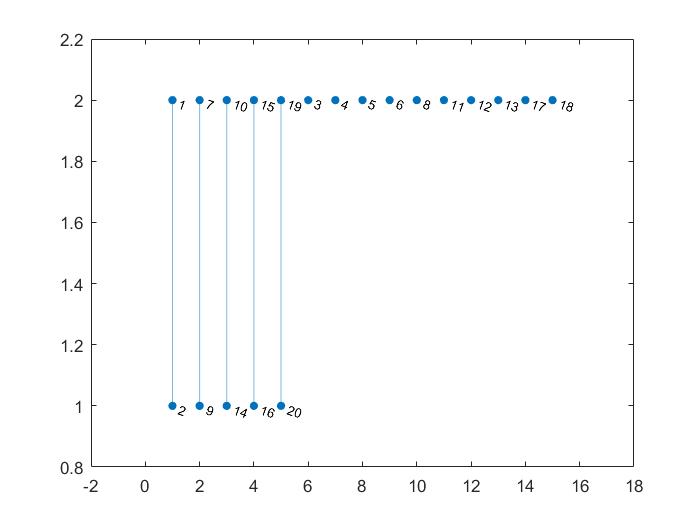}
		\caption{Designed unconnected topology from a ring(optimal)}
        \label{fig:unconnected_ring}
	\end{subfigure}
	\begin{subfigure}{.3\textwidth}
		\includegraphics[width=\textwidth]{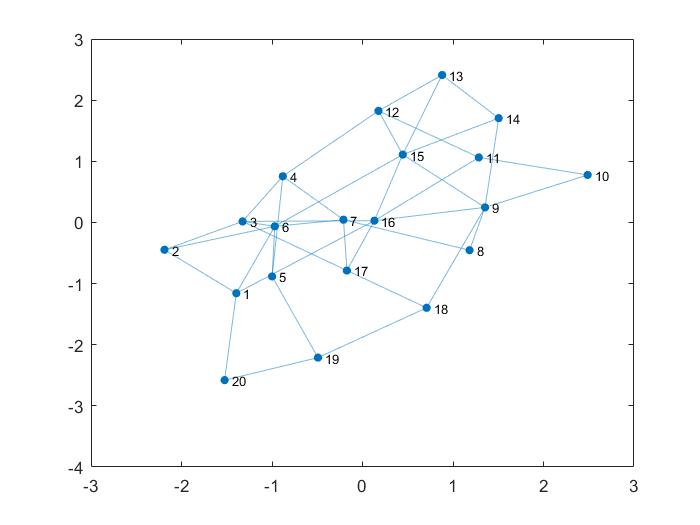}
		\caption{Designed connected topology from a ring}
        \label{fig:connected_ring}
	\end{subfigure}
\end{figure}

Even if it is mathematically acceptable, the isolation of the agents is a bit unrealistic and in many cases undesirable design. Subsequently, we add a simple linear constraint in the topology design problem demanding the designed topology to have at least 19 edges -the minimum edges needed to be connected. Interestingly, we obtain a connected topology Fig:\ref{fig:connected_ring}, which has 39 edges and differs from the initial one on 20 edges.

\subsubsection{4-lattice}

\begin{figure}[h!]
	\centering
	\begin{subfigure}{.45\textwidth}
        \includegraphics[width=\textwidth]{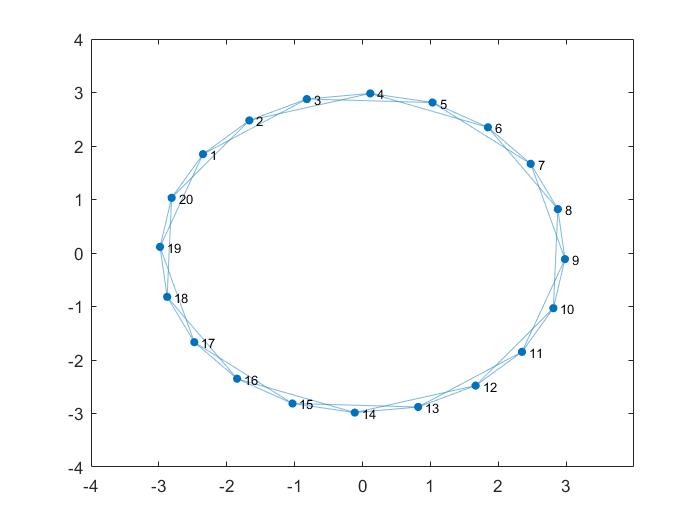}
        \caption{Initial 4-lattice topology}
        \label{fig:4_lattice}
    \end{subfigure}
	\begin{subfigure}{.45\textwidth}
		\includegraphics[width=\textwidth]{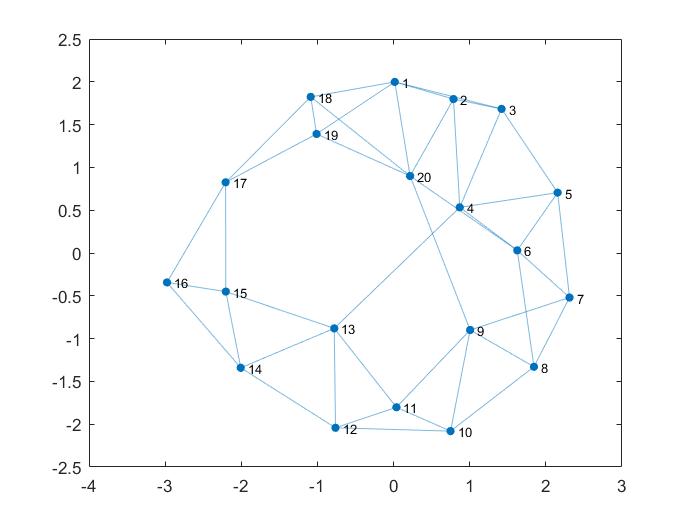}
		\caption{Designed topology from 4-lattice}
        \label{fig:4_lattice_final}
	\end{subfigure}
\end{figure}

For the 4-lattice Fig:\ref{fig:4_lattice} the parameters $g_i$ indicating the manipulative tendencies of the agents are chosen randomly from the interval $[0,20]$. This increase in the manipulation parameters shows from the beginning that the lattices have enhanced stability properties in comparison with the ring, as it is expected since they are more dense and well connected topologies. The 4-lattice depicted in Fig:\ref{fig:4_lattice} has 40 edges. Our design results in the topology Fig:\ref{fig:4_lattice_final} which has 43 edges and differs from the initial one on 5 edges.

\subsubsection{Star}

For the star topology Fig:\ref{fig:star} the parameters $g_i$ indicating the manipulative tendencies of the agents are chosen randomly from the interval $[0,20]$, except the one of the central node which is chosen much larger (here $g(1)=70$). That is because the star structure is a very robust one with respect to its stability properties, since the central node is very difficult to manipulate and to be manipulated as she has the most neighbors she could have. So, the parameters should be chosen large enough in order to arise instabilities on this initial topology. Moreover, the star graph has the least possible edges needed to be connected (19 edges), so it seems to be a very robust design for the number of its edges. That is the reason why our algorithm converges to an unconnected topology Fig:\ref{fig:unconnected_star} which is closer to the star topology than any connected stable one. It has only 3 edges and it differs form the initial topology on 18 edges. It shall be noted here that, as in the case of the ring, the unconnected designed topology is stable.

\begin{figure}[h!]
	\centering
	\begin{subfigure}{.3\textwidth}
        \includegraphics[width=\textwidth]{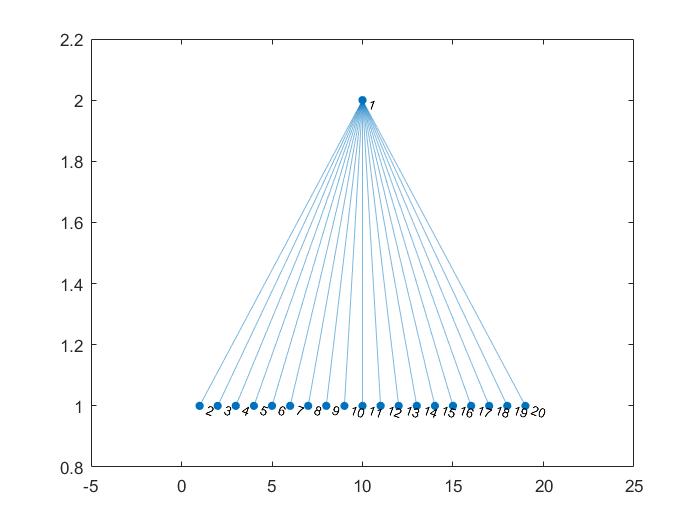}
        \caption{Initial star topology}
        \label{fig:star}
    \end{subfigure}
	\begin{subfigure}{.3\textwidth}
		\includegraphics[width=\textwidth]{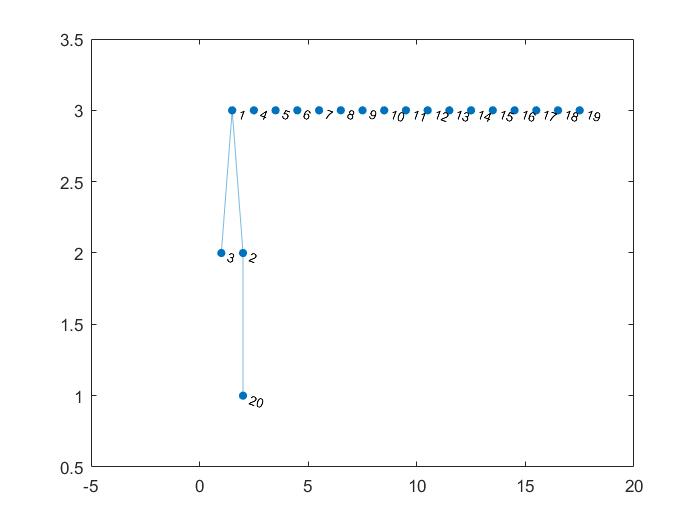}
		\caption{Designed unconnected topology from a star (optimal)}
        \label{fig:unconnected_star}
	\end{subfigure}
	\begin{subfigure}{.3\textwidth}
		\includegraphics[width=\textwidth]{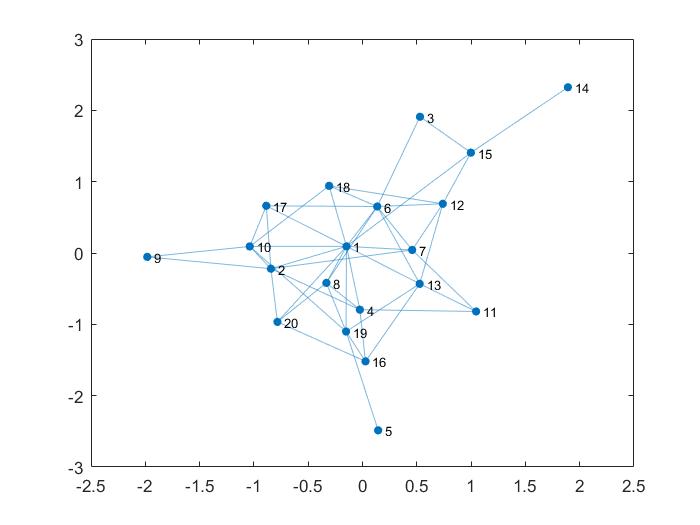}
		\caption{Designed connected topology from a star}
        \label{fig:connected_star}
	\end{subfigure}
\end{figure}

Subsequently, as in the case of the ring topology, we add an extra constraint for the topology to enhance a connected design and we derive the final topology depicted in Fig:\ref{fig:connected_star}. It has 47 edges and it differs from the initial one on 42 edges.

\textbf{Comments:} From the study of these special structures we deduce several interesting conclusions. At first, in the case of topologies with very few edges, such as a ring or a star graph, the isolation of some agents from the rest network is sometimes optimal as it effectively stops their manipulative activity. The fact is that such a design may not be acceptable by these agents and by the society. So, we add more constraints, which do not affect the difficulty of the problem, in order to avoid a design which may be optimal but inapplicable in social networks. Fortunately, since the increase of the agent's friends leads to the decrease of her ability to manipulate each one of them, as we deduce from the sufficient condition \eqref{suf}, it is guaranteed that there exists another topology with more edges than the initial, which satisfies the stability constraints and it is in fact a local minimum of our optimisation problem. We can also design this topology to be connected by adding more edges and not affecting its stability.

\section{Conclusion}
In this work we considered a social choice procedure as a repeated Nash game between social agents communicating over a social network. One contribution of this work is an enrichment of the model for social choice procedures proposed in \cite{Etesami3} by considering dynamically changing opinions and thus resulting in second order dynamics. However, our basic novelty is a new approach for the stabilization of these dynamics through the graph topology design, which results in an integer programming problem with a BMI constraint. Finally, we designed a proper genetic algorithm for this problem. Applying this algorithm to several graph topologies which had resulted in unstable dynamics we found that in the case of topologies with few edges, such as a ring or a star, the isolation of the manipulative agents is an optimal (or suboptimal) design, however it may not be socially preferable. Contrary to that, in the case of well-connected topologies the addition or the rewiring of a few edges can affect significantly the manipulative agents and result in the stabilisation of the dynamics.

\bibliography{bibliography-opinions}

\begin{thebibliography}{10}
\providecommand \doibase [0]{http://dx.doi.org/}%

\bibitem{de_groot}
DeGroot MH. Reaching a consensus. {\it Journal of the American Statistical
  Association} 1974\string; 69\string: 118--121.

\bibitem{Friedkin1}
Friedkin NE, Johnsen EC. Social influence and opinions. {\it Journal of
  Mathematical Sociology} 1990\string; 15(3-4)\string: 193--205.

\bibitem{Friedkin2}
Friedkin NE, Johnsen EC. Social influence networks and opinion change. {\it
  Advances in Group Processes} 1999\string; 16(1)\string: 1--29.

\bibitem{Sznajd}
Sznajd-Weron K, Sznajd J. Opinion evolution in closed community. {\it
  International Journal of Modern Physics C} 2000\string; 11(6)\string:
  1157--1165.

\bibitem{h&k1}
Hegselmann R, Krause U. Opinion dynamics and bounded confidence models,
  analysis and simulation. {\it Journal of Artifical Societies and Social
  Simulation (JASSS)} 2002\string; 5(3).

\bibitem{h&k2}
Hegselmann R, Krause U. Opinion Dynamics Driven by Various Ways of Averaging.
  {\it Computational Economics} 2005\string; 25\string: 381--405.

\bibitem{Fortunato1}
Fortunato S. The Krause Hegselmann consensus model with discrete opinions. {\it
  International Journal of Modern Physics C} 2004\string; 15(7)\string:
  1021--1029.

\bibitem{Krause}
Krause U. Compromise, consensus, and the iteration of means. {\it Elem. Math.}
  2009\string; 64\string: 1--8.

\bibitem{extremist}
Amblard F, Deffuant G. The role of network topology on extremism propagation
  with the relative agreement opinion dynamics. {\it Physica A} 2004\string;
  343\string: 725--738.

\bibitem{Deffuant}
Weisbucha G, Deffuant G, Amblarda F. Persuasion dynamics. {\it Physica
  A}\string; 353.

\bibitem{Gionis}
Gionis A, Terzi E, Tsaparas P. Opinion maximization in social networks. {\it
  Proceedings of the 2013 SIAM International Conference on Data Mining.}
  2013\string: 387--395.

\bibitem{stubborn}
Ghaderi J, Srikant R. Opinion dynamics in social networks with stubborn agents:
  Equilibrium and convergence rate. {\it Automatica} 2014\string; 50\string:
  3209--3215.

\bibitem{acemoglu2010spread}
Acemoglu D, Ozdaglar A, ParandehGheibi A. Spread of (mis) information in social
  networks. {\it Games and Economic Behavior} 2010\string; 70(2)\string:
  194--227.

\bibitem{forster2014trust}
Forster M, Mauleon A, Vannetelbosch VJ. Trust and manipulation in social
  networks.  2014.

\bibitem{Survey}
Ren W, Beard RW, Atkins EM. A Survey of Consensus Problems in Multi-agent
  Coordination. {\it 2005 American Control Conference, Portland, OR, USA} June
  8-10, 2005.

\bibitem{Friedkin_book}
Friedkin NE, Johnsen EC. {\it Social influence network theory}.
\newblock Cambridge Univercity Press .
\newblock 2011.

\bibitem{Etesami1}
Etesami SR, Bolouki S, Nedic A, Ba\c{s}ar T. Conformity versus manipulation in
  reputation systems. {\it IEEE 55th Conference on Decision and Control (CDC)}
  2016\string: 4451--4456.

\bibitem{Etesami2}
Etesami SR, Bolouki S, Ba\c{s}ar T, Nedic A. Evolution of Public Opinion under
  Conformist and Manipulative Behaviors. {\it Preprints of the 20th World
  Congress, IFAC, Toulouse, France} July 9-14, 2017.

\bibitem{Etesami3}
Etesami SR, Bolouki S, Nedic A, Ba\c{s}ar T, Poor HV. Influence of Conformist
  and Manipulative Behaviors on Public Opinion. {\it IEEE Transactions on
  Control of Network Systems} 2018.

\bibitem{chandrasekhar2012testing}
Chandrasekhar AG, Larreguy H, Xandri JP. Testing models of social learning on
  networks: Evidence from a framed field experiment. {\it Work. Pap., Mass.
  Inst. Technol., Cambridge, MA} 2012.

\bibitem{lapiere1934attitudes}
LaPiere RT. Attitudes vs. actions. {\it Social forces} 1934\string;
  13(2)\string: 230--237.

\bibitem{wicker1969attitudes}
Wicker AW. Attitudes versus actions: The relationship of verbal and overt
  behavioral responses to attitude objects. {\it Journal of Social issues}
  1969\string; 25(4)\string: 41--78.

\bibitem{glasman2006forming}
Glasman LR, Albarracin D. Forming attitudes that predict future behavior: A
  meta-analysis of the attitude-behavior relation.. {\it Psychological
  bulletin} 2006\string; 132(5)\string: 778.

\bibitem{marsh2005influence}
Marsh KL, Wallace HM. The Influence of Attitudes on Beliefs: Formation and
  Change..  2005.

\bibitem{von2007theory}
Von~Neumann J, Morgenstern O. {\it Theory of games and economic behavior
  (commemorative edition)}.
\newblock Princeton university press .
\newblock 2007.

\bibitem{Kordonis1}
Kordonis I, Charalampidis AC, Papavassilopoulos GP. Pretending in Dynamic
  Games, Alternative Outcomes and Application to Electricity Markets. {\it
  Dynamic Games and Applications} 2017\string: 1--30.

\bibitem{benkler2018network}
Benkler Y, Faris R, Roberts H. {\it Network propaganda: Manipulation,
  disinformation, and radicalization in American politics}.
\newblock Oxford University Press .
\newblock 2018.

\bibitem{aral2019protecting}
Aral S, Eckles D. Protecting elections from social media manipulation. {\it
  Science} 2019\string; 365(6456)\string: 858--861.

\bibitem{szwarcberg2015mobilizing}
Szwarcberg M. {\it Mobilizing poor voters: Machine politics, clientelism, and
  social networks in Argentina}.
\newblock Cambridge University Press .
\newblock 2015.

\bibitem{woolley2018computational}
Woolley SC, Howard PN. {\it Computational propaganda: political parties,
  politicians, and political manipulation on social media}.
\newblock Oxford University Press .
\newblock 2018.

\bibitem{Kordonis2}
Kordonis I, Papavassilopoulos GP. Network Design in the Presence of a Link
  Jammer: a Zero-Sum Game Formulation. {\it IFAC 2017 World Congress, Toulouse,
  France, The 20th World Congress of the International Federation of Automatic
  Control} 9-14 July 2017.

\bibitem{chen2007network}
Chen BK, Tobagi FA. Network topology design to optimize link and switching
  costs. In: IEEE. ; 2007\string: 2450--2456.

\bibitem{ferentinos2007adaptive}
Ferentinos KP, Tsiligiridis TA. Adaptive design optimization of wireless sensor
  networks using genetic algorithms. {\it Computer Networks} 2007\string;
  51(4)\string: 1031--1051.

\bibitem{kar2008sensor}
Kar S, Moura JM. Sensor networks with random links: Topology design for
  distributed consensus. {\it IEEE Transactions on Signal Processing}
  2008\string; 56(7)\string: 3315--3326.

\bibitem{gross2011optimized}
Gro{\ss} D, Stursberg O. Optimized distributed control and network topology
  design for interconnected systems. In: IEEE. ; 2011\string: 8112--8117.

\bibitem{khan2012fuzzy}
Khan SA, Engelbrecht AP. A fuzzy particle swarm optimization algorithm for
  computer communication network topology design. {\it Applied Intelligence}
  2012\string; 36(1)\string: 161--177.

\bibitem{saad2018multi}
Saad A, Khan SA, Mahmood A. A multi-objective evolutionary artificial bee
  colony algorithm for optimizing network topology design. {\it Swarm and
  Evolutionary Computation} 2018\string; 38\string: 187--201.

\bibitem{uavs}
Chapman A, Dai R, Mesbahi M. Network Topology Design for UAV Swarming with Wind
  Gusts. {\it AIAA Guidance, Navigation, and Control Conference, Portland,
  Oregon} 8 - 11 August 2011.

\bibitem{cyberphysical}
Khaitan SK, McCalley JD. Design Techniques and Applications of Cyber Physical
  Systems: A Survey. {\it IEEE Systems Journal} July 2014.

\bibitem{Kordonis3}
Kordonis I, Papavassilopoulos GP. Network Design for Fast Convergence to the
  Nash Equilibrium in a Class of Repeated Games. {\it 24th Mediterranean
  Conference on Control and Automation, Athens, Greece} June 21-24, 2016.

\bibitem{Safonov}
Mesbahi M, Safonov MG, Papavassilopoulos GP. Bilinearity and Complementarity in
  robust control.. {\it In L. El Ghaoui and S. Niculescu, editors, Advances in
  Linear Matrix Inequality Methods in Control} 2000\string: 269--292.

\bibitem{Tsitsiklis}
Olshevsky A, Tsitsiklis JN. Convrgence speed in distributed consensus and
  averaging. {\it Journal of Control and Optimization} 2009\string;
  48(1)\string: 33--55.

\bibitem{Polyak}
Polyak BT. {\it Introduction to Optimization}.
\newblock Optimization Software, Inc., Publications Division, New York .
\newblock 1987.

\bibitem{random}
Erd\H{o}s P, R\'{e}nyi A. On random graphs I. {\it Publicationes
  Mathematicae}\string; 6\string: 290--297.

\bibitem{vanantwerp2000tutorial}
VanAntwerp JG, Braatz RD. A tutorial on linear and bilinear matrix
  inequalities. {\it Journal of process control} 2000\string; 10(4)\string:
  363--385.

\bibitem{goldberg1989genetic}
Goldberg DE. Genetic algorithms in search. {\it Optimization, and
  MachineLearning} 1989.

\bibitem{kumar1995genetic}
Kumar A, Pathak RM, Gupta YP, Parsaei HR. A genetic algorithm for distributed
  system topology design. {\it Computers \& Industrial Engineering}
  1995\string; 28(3)\string: 659--670.

\end{thebibliography}

\end{document}